\pdfoutput=1
\documentclass[conference]{IEEEtran}
\IEEEoverridecommandlockouts
\usepackage{comment}
\includecomment{ISIT}
\excludecomment{full}

\usepackage{cite}
\usepackage{graphicx}
\usepackage{amsmath}
\usepackage{url}

\usepackage{amssymb,mathtools}
\usepackage{amsthm}
\usepackage{etoolbox} 
\usepackage{tabularx}
\usepackage{multirow}
\usepackage{booktabs}
\usepackage{braket}
\usepackage{capt-of}
\usepackage[ruled,vlined,linesnumbered]{algorithm2e}

\theoremstyle{definition}
\newtheorem{definition}{Definition}
\newtheorem{example}{Example}
\theoremstyle{plain}
\newtheorem{theorem}{Theorem}

\newtheorem{lemma}{Lemma}
\newtheorem{corollary}{Corollary}

\AfterEndEnvironment{definition}{\noindent\ignorespaces}
\AfterEndEnvironment{example}{\noindent\ignorespaces}
\AfterEndEnvironment{theorem}{\noindent\ignorespaces}
\AfterEndEnvironment{proposition}{\noindent\ignorespaces}
\AfterEndEnvironment{lemma}{\noindent\ignorespaces}
\AfterEndEnvironment{corollary}{\noindent\ignorespaces}
\AfterEndEnvironment{proof}{\noindent\ignorespaces}

\makeatletter
\renewenvironment{proof}[1][\proofname]{\par
  \pushQED{\qed}%
  \normalfont
  \topsep0pt \partopsep0pt 
  \trivlist
  \item[\hskip\labelsep\itshape#1\@addpunct{.}]\ignorespaces
  }{%
  \popQED\endtrivlist\@endpefalse
  \addvspace{6pt plus 6pt} 
  }
\makeatother

\usepackage{color}

\newcommand{\q}{\hspace{10pt}}

\newcommand{\Sc}{\mathcal{S}}
\newcommand{\Pc}{\mathcal{P}}

\newcommand{\sen}{[n]}

\newcommand{\R}{\mathbb{R}}
\newcommand{\Z}{\mathbb{Z}}

\renewcommand{\vec}[1]{\boldsymbol{#1}}
\renewcommand{\phi}{\varphi}

\newcommand{\KL}{D_{\mathrm{KL}}}
\newcommand{\MI}{\mathrm{MI}}
\newcommand{\RI}[4]{\mathrm{RI}_{#3 \to #4}(#1, #2)}

\newcommand{\defeq}{\coloneqq}

\begin{document}
%
\title{Information Decomposition on Structured Space\thanks{A preprint is available at http://arxiv.org/abs/1601.05533.}}

\author{\IEEEauthorblockN{Mahito Sugiyama}
\IEEEauthorblockA{ISIR, Osaka University\\
JST PRESTO\\
mahito@ar.sanken.osaka-u.ac.jp}
\and
\IEEEauthorblockN{Hiroyuki Nakahara}
\IEEEauthorblockA{RIKEN Brain Science Institute\\
hiro@brain.riken.jp}
\and
\IEEEauthorblockN{Koji Tsuda}
\IEEEauthorblockA{Graduate School of Frontier Sciences\\
The University of Tokyo\\
tsuda@k.u-tokyo.ac.jp}}



\maketitle

\begin{abstract}
We build information geometry for a partially ordered set of variables and define the  
orthogonal decomposition of information theoretic quantities.
The natural connection between information geometry and order theory leads to efficient decomposition algorithms.
This generalization of Amari's seminal work on hierarchical decomposition of probability distributions 
on event combinations enables us to analyze high-order statistical interactions
arising in neuroscience, biology, and machine learning.
 \end{abstract}


\begin{full}
\end{full}

%
\IEEEpeerreviewmaketitle

\section{Introduction}\label{sec:intro}
Let $e_1, e_2, \ldots, e_n$ denote the set of events. 
All combinations of events are regarded as a partially ordered set and form a complete hierarchy (Figure~\ref{figure:hierarchy}\textbf{a}).
Amari introduced the orthogonal decomposition of probability distributions defined on the complete hierarchy of events~\cite{Amari01}.
That method provided a theoretical foundation with which to analyze the higher-order
interactions in a wide variety of applications, such as firing
patterns of neurons~\cite{Nakahara02,Nakahara06}, gene
interactions~\cite{Nakahara03}, and word associations in
documents~\cite{Hou13}.
However, in many applications the hierarchy is often incomplete,
because some event combinations can never occur (Figure~\ref{figure:hierarchy}\textbf{b}).
For example, if $e_1$ indicates a person being male and $e_2$ indicates a person having ovarian cancer,
the combination of $e_1$ and $e_2$ can never occur.
Incomplete hierarchies can also result from a lack of data~\cite{Ganmor11}.

We define information geometric dual coordinates on a \emph{partially ordered set}, or a \emph{poset}. 
They lead to an efficient algorithm for decomposing Kullback--Leibler divergence and entropy in an incomplete hierarchy.
Our method can be used to isolate the contribution of each event combination and assess its
statistical significance~\cite{Nakahara02}. 
From a theoretical viewpoint, our method offers a previously unexplored link between order theory and information geometry.

The remainder of this paper is organized as follows.
 Section~\ref{sec:probPosets} introduces a dually flat manifold on a poset.
 In Section~\ref{subsec:coordinate}, we show that, given a poset we introduce, the manifold of probability distributions will always have the same dually flat structure as that of the exponential family of the original variable set (Equations~\eqref{eq:loglinear} and~\eqref{eq:eta}).
 In Section~\ref{subsec:mixed}, we present an efficient algorithm to decompose information on a poset (Algorithms~\ref{alg:mixedsingle},~\ref{alg:mixedmulti} and Theorem~\ref{theorem:pythagoras}).
\begin{ISIT}
 As a representative application, in Section~\ref{sec:learning}, we show that our algorithm can efficiently isolate information of arbitrary order interactions of events.
 We summarize and conclude the paper in Section~\ref{sec:conclusion}.
\end{ISIT}
\begin{full}
 We discuss a metric on the manifold in Section~\ref{subsec:distance} and the statistical significance in Section~\ref{subsec:significance}.
 Section~\ref{sec:MIdecompose} is devoted to introducing nonnegative decomposition of the mutual information (Theorem~\ref{theorem:MIKLdecompose}) with proposing a new notion of the refined mutual information, which is a generalized concept of mutual information.
 In contrast to the above two sections in which a poset is assumed to be given beforehand, Section~\ref{sec:learning} discusses how to construct such structured spaces from data.
 As a representative application, we show in Section~\ref{subsec:binary} that our algorithm can efficiently isolate information of arbitrary order interactions of events.
 We summarize and conclude the paper in Section~\ref{sec:conclusion}.
\end{full}

\begin{figure}[t]
 \centering
 \includegraphics[width=\linewidth]{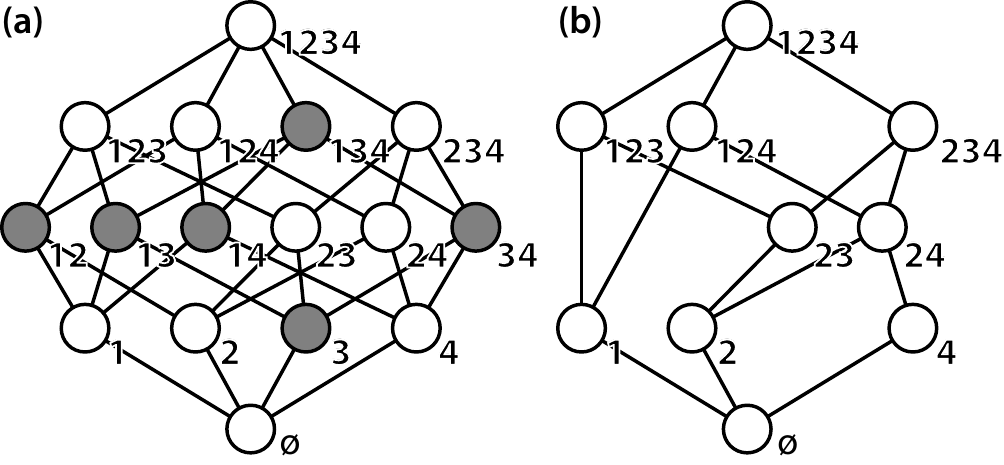}
 \caption{Hierarchy of combinations of four events $e_1, e_2, e_3$, and $e_4$. Numbers denote corresponding events.
 \textbf{(a)} The complete hierarchy of combinations of events. \textbf{(b)} Incomplete hierarchy by removing gray combinations in \textbf{(a)}.}
 \label{figure:hierarchy}
\end{figure}

\begin{full}
 \begin{table}[t]
 \caption{Notation.}
 \label{table:notation}
 \begin{tabularx}{\linewidth}{Xl}
  \toprule
  $S$ & Poset\\
  $\bot$ & Bottom element of $S$, i.e., $\bot \le x$ for all $x \in S$\\
  $S^+$ & $S \setminus \{\bot\}$\\
  $s, x, y, z$ & Objects in $S$\\
  $I, J, K$ & Subsets of $S$\\
  ${\downarrow}I$, ${\uparrow}I$ & Lower and upper sets, ${\downarrow}x = {\downarrow}\{x\}$ and ${\uparrow}x = {\uparrow}\{x\}$\\
  $y \lessdot x$ & $x$ covers $y$, i.e., $y < x$ and $y \le z < x$ implies $y = z$\\
  $X, Y$ & Discrete random variable,\\
  & $X$ is always with the alphabet $S$\\
  $p, q, r$ & Probability distribution on $S$\\
  $i, j, k, l$ & Natural number\\
  $N$ & Sample size\\
  $\vec{\Sc}$ & Manifold on $S$;\\
  & $\vec{\Sc} = \{\,p \mid \forall x \in S, p(x) > 0\ \text{and} \sum_{x \in S} p(x) = 1\,\}$\\
  $\theta$, $\eta$ & Coordinate systems of $\vec{\Sc}$\\
  & (see Equations~\eqref{eq:theta_estimation}, \eqref{eq:eta})\\
  $\theta_p$, $\eta_p$ & $\theta$- and $\eta$-coordinates of $p$\\
  $E_I$ & Submanifold of $\vec{\Sc}$;\\
  & $E_I(p) = \{\,q \in \vec{\Sc} \mid \theta_{q}(x) = \theta_p(x) \text{ for all } x \in I\,\}$\\
  $M_I$ & Submanifold of $\vec{\Sc}$;\\
  & $M_I(p) = \{\,q \in \vec{\Sc} \mid \eta_{q}(x) = \eta_p(x) \text{ for all } x \in I\,\}$\\
  $\KL(p, q)$ & KL divergence from $p$ to $q$\\
  $H(X)$ & Entropy of $X$\\
  $v$ & Subvaluation (see Equation~\eqref{eq:valdef})\\
  $\MI(X, Y)$ & Mutual information between $X$ and $Y$\\
  $\RI{X}{Y}{I}{J}$ & Refined mutual information (see Definition~\ref{def:RMI})\\
  $\sigma$ & Threshold for the frequency of samples\\
  \bottomrule
 \end{tabularx}
 \end{table}
\end{full}

\section{Dually Flat Manifold on Posets}\label{sec:probPosets}
Suppose that $X$ is a discrete random variable and $p(x) = \Pr(X = x)$ with $x \in S$ is a probability mass function on a finite set $S$.
In information geometry~\cite{Amari01,Amari07}, each distribution is treated as a mapping $p {:} S \to \R$ and
the set of all probability distributions is understood to be a $(|S| - 1)$-dimensional manifold
\begin{ISIT}
 $\vec{\Sc} = \{\,p \mid p(x) > 0 \text{ for all } x \in S,\hspace*{5pt} \sum_{x \in S} p(x) = 1\,\}$,
\end{ISIT}
\begin{full}
 \begin{align*}
 \vec{\Sc} = \Set{p | p(x) > 0 \text{ for all } x \in S,\hspace*{5pt} \sum_{x \in S} p(x) = 1},
 \end{align*}
\end{full}
where probabilities form a coordinate system of $\vec{\Sc}$, called the \emph{$p$-coordinate system}.
Information geometry gives us two more coordinate systems of $\vec{\Sc}$, the $\theta$-coordinate system and the $\eta$-coordinate system,
which are known to be dually orthogonal and key to decomposing KL divergence via the mixed coordinate system of $\theta$ and $\eta$.
We introduce such two coordinates in Section~\ref{subsec:coordinate} and show decomposition of KL divergence in Section~\ref{subsec:mixed}.

We consider the case where $S$ is a partially ordered set, or a \emph{poset}, which is one of the most fundamental structured space in computer science and mathematics.
A \emph{partial order} ``$\le$'' satisfies the following three properties: for all $x, y, z \in S$, (1) $x \le x$ (reflexivity), (2) $x \le y$, $y \le x \Rightarrow x = y$ (antisymmetry), and (3) $x \le y$, $y \le z \Rightarrow x \le z$ (transitivity).
Throughout the paper, we assume that $S$ is always finite and includes the bottom element $\bot \in S$; that is, $\bot \le x$ for all $x \in S$.
We write the set $S \setminus \{\bot\}$ by $S^+$.

For a subset $I \subseteq S$, we denote a lower set ${\downarrow}I = \{\,x \in S \mid x \le s \text{ for some } s \in I\,\}$, an upper set ${\uparrow}I = \{\,x \in S \mid x \ge s \text{ for some } s \in I\,\}$, and ${\downarrow}x = {\downarrow}\{x\}$, ${\uparrow}x = {\uparrow}\{x\}$ for each $x \in S$.
In order theory, ${\downarrow}x$ is called the \emph{principal ideal} for $x$ and ${\uparrow}x$ is called the \emph{principal filter} for $x$~\cite{Davey02,Gierz03}, which are known to be fundamental mathematical objects in posets.

\subsection{$\theta$- and $\eta$-coordinate Systems}\label{subsec:coordinate}
Let us first introduce the $\theta$-\emph{coordinate system} of the manifold $\vec{\Sc}$, which is realized as a mapping $\theta {:} S \to \R$.
In the exponential family, $\theta$ is known to be the natural parameter, which is treated as an $n$-dimensional vector $\vec{\theta} = (\theta^1, \theta^2, \dots, \theta^n)$ and the distribution is in the form of
\begin{align}
 \label{eq:exponential_org}
 p(x; \vec{\theta}) = \exp\left(\,\sum_{i = 1}^{n} \theta^i F_i(x) - \psi(\vec{\theta})\,\right)
\end{align}
with $n$ functions $F_1, \dots, F_{n}$ and a normalizer $\psi(\vec{\theta})$~\cite{Amari07}.
This is re-written as
\begin{align}
 \label{eq:exponential_poset}
 p(x; \theta) = \exp\left(\,\sum_{s \in S^+} \theta(s) F_{s}(x) - \psi(\theta)\,\right)
 \end{align}
with $n = |S^+|$ in our setting, where there exists a one-to-one indexing mapping $\omega {:} S^+ \to \{1, 2, \dots, n\}$ such that $\theta(s)$ and $F_s$ correspond to $\theta^{\omega(s)}$ and $F_{\omega(s)}$ in Equation~\eqref{eq:exponential_org}, respectively.

Given a poset $S$, we propose to define $F_s(x)$ as
\begin{align*}
 F_s(x) &= \left\{
 \begin{array}{ll}
  1 & \text{if } s \le x, \\
  0 & \text{otherwise}
 \end{array}
 \right.
 \ \text{and}\ \psi(\theta) = -\log p(\bot).
\end{align*}
Interestingly, from Equation~\eqref{eq:exponential_poset}, we obtain the expansion of $\log p(x)$ as the sum of $\theta(s)$ of lower elements $s \le x$ in $S$:
\begin{align}
 \label{eq:loglinear}
 \log p(x) = \sum_{s \le x} \theta(s).
\end{align}
Note that this equation can be viewed as a generalization of the well-known log-linear model:
\begin{align*}
 \log p(\vec{x}) = &\sum_{i} \theta^i x^i + \sum_{i < j} \theta^{ij} x^i x^j + \sum_{i < j < k} \theta^{ijk} x^i x^j x^k \nonumber\\
 &+ \dots + \theta^{1\dots n} x^1 \dots x^n - \psi
\end{align*}
for $n$-dimensional binary vector $\vec{x} = (x^1, \dots, x^n) \in \{0, 1\}^n$.

Thus, given a probability distribution $p \in \vec{\Sc}$, the $\theta$-coordinate system $\theta {:} S \to \R$ is recursively computed as
\begin{align}
 \label{eq:theta_estimation}
 \theta(x) = \log p(x) - \displaystyle\sum_{s < x} \theta(s)
\end{align}
starting from the bottom $\theta(\bot) = \log p(\bot)$.

\begin{figure}[t]
 \centering
 \includegraphics[width=.7\linewidth]{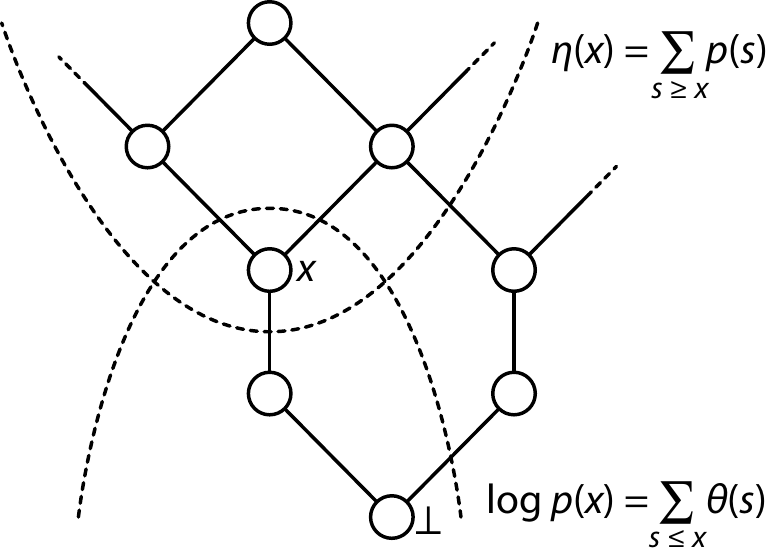}
 \caption{$p(x)$, $\theta(x)$, and $\eta(x)$ on poset.}
 \label{figure:poset}
\end{figure}

In information geometry~\cite{Amari01}, the natural parameter $\theta$ of the exponential family is known to be the $e$-affine coordinate of the $e$-flat manifold $\vec{\Sc}$,
which means that our formulation of $\theta$ in Equation~\eqref{eq:theta_estimation} is the $e$-affine coordinate.
The $m$-affine coordinate $\eta {:} S \to \R$, an alternative coordinate system that introduces the duality to $\vec{\Sc}$, is given as the expectation of the parameter $F_s(x)$ for each $s \in S$.
In our case $\eta$ is given as follows:
\begin{align}
 \label{eq:eta}
 \eta(s) = \mathbb{E}[F_s(x)] = \sum_{x \ge s} p(x) = \Pr(X \ge s).
\end{align}
Relationships of $p$, $\theta$, and $\eta$ are illustrated in Figure~\ref{figure:poset}.

The two coordinate systems $\theta$ and $\eta$ are connected with each other by the Legendre transformation.
The remarkable property is that $\theta$ and  $\eta$ are \emph{dually orthogonal}:
\begin{align}
 \label{eq:orthogonal}
 \mathbb{E}\left[\,\frac{\partial}{\partial\theta(s)}\log p(x; \theta) \frac{\partial}{\partial\eta(s')}\log p(x; \eta)\,\right] = \delta(s, s')
\end{align}
for every $s, s' \in S^{+}$ with the Kronecker delta $\delta$ such that $\delta(s, s') = 1$ if $s = s'$ and $\delta(s, s') = 0$ otherwise~\cite{Amari01}.
This property is essential to construct a mixed coordinate system of $\theta$ and $\eta$ in the next subsection.

Our finding connects two fundamental areas, information geometry and order theory, that have been independently studied to date.
Given the $\theta$-coordinate, our result means that the $p$-coordinate is generated from the set of \emph{principal ideals} and the $\eta$-coordinate is generated from the set of \emph{principal filters}.
More specifically, let $f(I, \theta) = \exp(\,\sum_{s \in I} \theta(s)\,)$ and $g(I, p) = \sum_{s \in I} p(s)$ for every $I \subseteq S$.
For each $x \in S$, we have $p(x) = f({\uparrow}x, \theta)$ with the principal ideal ${\uparrow}x$ for $x$ and $\eta(x) = g({\downarrow}x, p)$ with the principal filter ${\downarrow}x$.

\begin{full}
\begin{algorithm}[t]
 \begin{small}
 \SetKwInOut{Input}{input}
 \SetKwInOut{Output}{output}
 \SetFuncSty{textrm}
 \SetCommentSty{textrm}
 \SetKwFunction{ComputeThetaAll}{{\scshape ComputeThetaAll}}
 \SetKwFunction{ComputeTheta}{{\scshape ComputeTheta}}
 \SetKwFunction{AggregateTheta}{{\scshape AggregateTheta}}
 \SetKwProg{myfunc}{}{}{}

 \myfunc{\ComputeThetaAll{$S$}}{
 Let $x_1, x_2, \dots, x_{|S|}$ be the topological ordering of $S$\\
 \ForEach(\tcp*[f]{\makebox[130pt][l]{Initializing flags indicating}}){$s \in S$ }{
 $f(s) \gets 0$ \tcp*[f]{\makebox[130pt][l]{whether objects are already visited}}
 }
 \For{$i \gets 1$ \KwTo $|S|$}{
 $\theta(x_i) \gets {}$\ComputeTheta{$x_i$}
 }
 }
 \myfunc{\ComputeTheta{$x$}}{
 $\theta(x) \gets \log p(x)$\\
 \ForEach{$s \lessdot x$}{
 $\theta(x) \gets \theta(x) - {}$\AggregateTheta{$s, \omega(x)$}\\
 }
 \Return $\theta(x)$\\
 }
 \myfunc{\AggregateTheta{$x, i$}}{
 \eIf{$f(x) = i$}{
 $\theta_{\text{sum}} \gets 0$
 }{
 $f(x) \gets i$\\
 $\theta_{\text{sum}} \gets \theta(x)$\\
 \ForEach{$s \lessdot x$}{
 $\theta_{\text{sum}} \gets \theta_{\text{sum}} + {}$\AggregateTheta{$s, i$}
 }
 }
 \Return $\theta_{\text{sum}}$
  }
 \caption{Compute $\theta$ from $p$}
 \label{alg:theta}
 \end{small}
\end{algorithm}

We present efficient algorithms that compute $\theta$ and $\eta$ from $p$.
The function {\scshape ComputeTheta}$(x)$ in Algorithm~\ref{alg:theta} recursively computes $\theta(x)$ from $\theta(s)$ with $s < x$, where we assume that, in our data structure, each pair $x, y \in S$ is connected with each other if $x \lessdot y$ or $y \lessdot x$.
The time complexity of computing each $\theta(x)$ is $O(|{\downarrow}\,x|)$ as the algorithm visits every $s < x$.
To compute all $\theta(x)$, first $S$ should be topologically sorted according to the order $\le$, resulting in a sequence $x_1, x_2, \dots, x_{|S|}$ such that $S = \{x_1, \dots, x_{|S|}\}$ and $x_i \le x_j$ for all $i \le j$ if they are comparable.
Then $\theta(x_1), \theta(x_2), \dots, \theta(x_{|S|})$ can be computed one after another using {\scshape ComputeTheta}$(x)$.
The time complexity of computing all $\theta(x)$, the function {\scshape ComputeThetaAll}$(S)$ in Algorithm~\ref{alg:theta}, is $O(\sum_{x \in S} |{\downarrow}\,x|) \le O(|S|^2)$.

In a similar way, the function {\scshape ComputeEta}$(x)$ in Algorithm~\ref{alg:eta} computes each $\eta(x)$ and {\scshape ComputeEtaAll}$(S)$ computes $\eta(x)$ for all $x \in S$ by tracing the set $S$ in reversed topological order.
The time complexity of computing $\eta(x)$ is $O(|{\uparrow}\,x|)$ and that of computing all $\eta(x)$ is $O(\sum_{x \in S} |{\uparrow}\,x|) \le O(|S|^2)$.

\begin{algorithm}[t]
 \begin{small}
 \SetFuncSty{textrm}
 \SetCommentSty{textrm}
 \SetKwFunction{ComputeEtaAll}{{\scshape ComputeEtaAll}}
 \SetKwFunction{ComputeEta}{{\scshape ComputeEta}}
 \SetKwFunction{AggregateEta}{{\scshape AggregateEta}}
 \SetKwProg{myfunc}{}{}{}

 \myfunc{\ComputeEtaAll{$S$}}{
 Let $x_1, x_2, \dots, x_{|S|}$ be the topological ordering of $S$\\
 \ForEach(\tcp*[f]{\makebox[130pt][l]{Initializing flags indicating}}){$s \in S$ }{
 $f(s) \gets 0$ \tcp*[f]{\makebox[130pt][l]{whether objects are already visited}}
 }
 \For{$i \gets |S|$ \KwTo $1$}{
 $\eta(x_i) \gets {}$\ComputeEta{$x_i$}
 }
 }
 \myfunc{\ComputeEta{$x$}}{
 $\eta(x) \gets p(x)$\\
 \ForEach{$s \gtrdot x$}{
 $\eta(x) \gets \eta(x) + {}$\AggregateEta{$s, \omega(x)$}
 }
 \Return $\eta(x)$\\
 }
 \myfunc{\AggregateEta{$x, i$}}{
 \eIf{$f(x) = i$}{
 $\eta_{\text{sum}} \gets 0$
 }{
 $f(x) \gets i$\\
 $\eta_{\text{sum}} \gets p(x)$\\
 \ForEach{$s \gtrdot x$}{
 $\eta_{\text{sum}} \gets \eta_{\text{sum}} + {}$\AggregateEta{$s, i$}
 }
 }
 \Return $\eta_{\text{sum}}$
 }
 \caption{Compute $\eta$ from $p$}
 \label{alg:eta}
 \end{small}
\end{algorithm}
\end{full}

\subsection{Information Decomposition via Mixed Coordinate System}\label{subsec:mixed}
We introduce the mixed coordinate system of $\theta$ and $\eta$~\cite{Amari01} on a poset, the key tool to analyze distributions on $S$.
The mixed coordinate system $\xi_I : S^+ \to \R$ with respect to a subset $I \subseteq S^+$ is a coordinate system of $\vec{\Sc}$ such that
\begin{align*}
 \xi_I(x) \defeq \left\{
 \begin{array}{ll}
  \eta(x)   & \text{if } x \in S^+ \setminus I, \\
  \theta(x) & \text{if } x \in I.
 \end{array}
 \right.
\end{align*}
Using the system, we can blend two distributions $p$ and $q$: The \emph{mixed distribution} of a pair of distributions $(p, q)$ with respect to $I \subseteq S^+$ is the distribution $r \in \vec{\Sc}$ such that
 \begin{align*}
 \left\{
 \begin{array}{ll}
  \eta_{r}(x) = \eta_p(x)     & \text{if } x \in S^+ \setminus I, \\
  \theta_{r}(x) = \theta_q(x) & \text{if } x \in I,
 \end{array}
 \right.
 \end{align*}
and $r(\bot) = 1 - \sum_{s \in S^+} r(x)$,
where we write $\theta$- and $\eta$-coordinates corresponding to $p$ by $\theta_p$ and $\eta_p$, respectively, to clarify that $p$, $\theta_p$, and $\eta_p$ are the same point in $\vec{\Sc}$.
Due to the orthogonality of $\theta$ and $\eta$ in Equation~\eqref{eq:orthogonal}, this distribution is always unique and well-defined.

\begin{figure}[t]
 \centering
 \includegraphics[width=.65\linewidth]{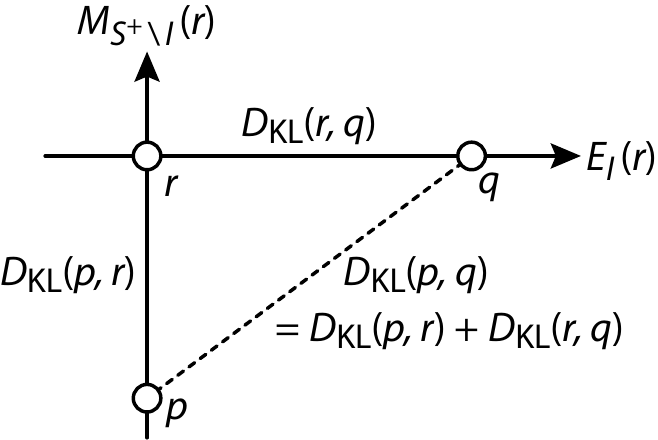}
 \caption{Pythagoras theorem (Theorem~\ref{theorem:pythagoras}).}
 \label{figure:pythagoras}
\end{figure}

Here we show decomposition of the Kullback--Leibler (KL) divergence between two probability distributions $p, q$:
\begin{align}
 \label{eq:kl}
 \KL(p, q) = \sum_{x \in S} p(x) \log\frac{p(x)}{q(x)}
\end{align}
using their mixed distribution $r$.
\begin{theorem}[Pythagoras theorem]
 \label{theorem:pythagoras}
 Given two distributions $p, q \in \vec{\Sc}$ and $I \subseteq S^+$.
 For a mixed distribution $r$ of $(p, q)$ and $r'$ of $(q, p)$ with respect to $I$,
 \begin{align}
  \KL(p, q) &= \KL(p, r) + \KL(r, q),\\
  \KL(q, p) &= \KL(q, r') + \KL(r', p).
 \end{align}
 \end{theorem}
\begin{proof}
 We can directly use Theorem~3 in~\cite{Amari01}, which shows that $\KL(p, q) = \KL(p, r) + \KL(r, q)$ holds if $m$-geodesic connecting $p$ and $r$ is orthogonal at $r$ to the $e$-geodesic connecting $r$ and $q$.
 Let two submanifolds $E_I(p)$ and $M_I(p)$ of $\vec{\Sc}$ be
\begin{align*}
 E_I(r) &\defeq \left\{\,\nu \in \vec{\Sc} \mid \theta_{\nu}(x) = \theta_r(x) \text{ for all } x \in I\,\right\},\\
 M_{S^+ \setminus I}(r) &\defeq \left\{\,\nu \in \vec{\Sc} \mid \eta_{\nu}(x) = \eta_r(x) \text{ for all } x \in S^+ \setminus I\,\right\}.
\end{align*}
 Since $E_{I}(r)$ and $M_{S^+ \setminus I}(r)$ are complementary and orthogonally intersect at $r$ from Equation~\eqref{eq:orthogonal},
 connection of $p$ and $r$ (resp.\ $r$ to $q$) is clearly $m$-geodesic (resp.\ $e$-geodesic; Figure~\ref{figure:pythagoras}).
 Therefore $\KL(p, q) = \KL(p, r) + \KL(r, q)$ follows.
 The second equation $\KL(q, p) = \KL(q, r') + \KL(r', p)$ can be proved in the same way.
\end{proof}
Moreover, for a hierarchical collection $\{I_0, I_1, \dots, I_k\}$ of subsets of $S$ such that $\emptyset = I_0 \subseteq I_1 \subseteq \dots \subseteq I_{k} = S^+$,
\begin{align}
 \label{eq:hdocmp}
  \KL(p, q) = \sum_{i = 1}^{k} \KL(r_{i - 1}, r_{i}),
 \end{align}
where $r_i$ is the mixed distribution of $(p, q)$ with respect to $I_i$ for each $i \in \{0, 1, \dots, k\}$.
Note that $r_0 = p$ and $r_k = q$.

Let $p_0$ be the uniform distribution such that $p_0(x) = 1/|S|$ for all $x \in S$, which is the origin of the $\theta$-coordinate because $\theta_{p_0}(x) = 0$ for all $x \in S^+$.
Since for the entropy $H(X)$ with a probability distribution $p$
\begin{align*}
 H(X)  = - \sum_{x \in S} p(x) \log p(x) = \log |S| - \KL(p, p_0)
\end{align*}
holds and $\log |S|$ is a constant,
entropy decomposition is achieved by our KL divergence decomposition in Theorem~\ref{theorem:pythagoras}:
\begin{align*}
 H(X)  = \log |S| - \left(\,\KL(p, r) + \KL(r, p_0)\,\right),
\end{align*}
where $r$ is the mixed distribution of $(p, p_0)$ with respect to $I \subseteq S^+$.
The entropy is decomposed into the contribution $\KL(p, r)$ of elements in $I$ and the other $\KL(r, p_0)$.
We can therefore obtain the \emph{information gain} for every subset $I \in S$ as the KL divergence $\KL(p, p_{I})$, where $p_I$ is the mixed distribution of $(p, p_0)$ with respect to $I$.


\subsection{Computation of Mixed Distributions}\label{subsec:KLdecompose}

\begin{algorithm}[t]
 \begin{small}
 \SetKwInOut{Input}{input}
 \SetKwInOut{Output}{output}
 \SetFuncSty{textrm}
 \SetCommentSty{textrm}
 \SetKwFunction{ComputeMixedSingle}{{\scshape ComputeMixedSingle}}
 \SetKwFunction{ComputeThetaSingle}{{\scshape ComputeThetaSingle}}
 \SetKwFunction{ComputePSingle}{{\scshape ComputePSingle}}
 \SetKwFunction{ComputeTheta}{{\scshape ComputeTheta}}
 \SetKwFunction{AggregateDiff}{{\scshape AggregateDiff}}
 \SetKwProg{myfunc}{}{}{}

 \myfunc{\ComputeMixedSingle{$x^*$}}{
 Let $x_1, \dots, x_{m} = x^*$ be the topological ordering of ${\downarrow}x^*$\\
 Search $r(x^*)$ that gives $\theta_r(x^*) = \theta_q(x^*)$ using \ComputeThetaSingle{$x^*, r(x^*)$} (e.g.\,bisection method)\\
 \Return $r$
 }
 \myfunc{\ComputeThetaSingle{$x^*, r(x^*)$}}{
 \ForEach(\tcp*[f]{\makebox[130pt][l]{Initializing flags indicating}}){$s \in {\downarrow}x^*$ }{
 $f(s) \gets 0$ \tcp*[f]{\makebox[130pt][l]{whether objects are already visited}}
 }
 \For{$i \gets m - 1$ \KwTo $1$}{
 $r(x_i) \gets {}$\ComputePSingle{$x_i$}
 }
 \For{$i \gets 1$ \KwTo $m$}{
 Compute $\theta_r(x_i)$
 }
 \Return $\theta_r(x^*)$\\
 }
 \myfunc{\ComputePSingle{$x$}}{
 $r(x) \gets p(x)$\\
 \ForEach(\tcp*[f]{\makebox[135pt][l]{$x \lessdot s$ iff $x < s$, $x \le y < s \Rightarrow x = y$}}){$s \gtrdot x$}{
 $r(x) \gets r(x) + {}$\AggregateDiff{$s, \omega(x)$}
 \tcp*[f]{\makebox[130pt][l]{$\omega(x)$: index of $x$}}
 }
 \Return $r(x)$\\
 }
 \myfunc{\AggregateDiff{$x, i$}}{
 \eIf{$f(x) = i$ or $x \not\le x^*$}{
 $p_{\text{diff}} \gets 0$
 }{
 $f(x) \gets i$\\
 $p_{\text{diff}} \gets p(x) - r(x)$\\
 \ForEach{$s \gtrdot x$}{
 $p_{\text{diff}} \gets p_{\text{diff}} + {}$\AggregateDiff{$s, i$}
 }
 }
 \Return $p_{\text{diff}}$
 }
 \caption{Compute the mixed distribution $r$ of $(p, q)$ with respect to a singleton $I = \{x^*\}$}
 \label{alg:mixedsingle}
 \end{small}
\end{algorithm}

\vspace*{-5pt}
\begin{algorithm}[t]
 \begin{small}
 \SetKwInOut{Input}{input}
 \SetKwInOut{Output}{output}
 \SetFuncSty{textrm}
 \SetCommentSty{textrm}
 \SetKwFunction{ComputeMixedMulti}{{\scshape ComputeMixedMulti}}
 \SetKwFunction{ComputeMixedSingle}{{\scshape ComputeMixedSingle}}
 \SetKwProg{myfunc}{}{}{}

 \myfunc{\ComputeMixedMulti{$I$}}{
 \Repeat{convergence of $r$}{
 \For{$x^* \in I$}{
  \ComputeMixedSingle{$x^*$}%
  \tcp*[f]{Algorithm~\ref{alg:mixedsingle}}
 }
 }
 }
 \caption{Compute the mixed distribution $r$ of $(p, q)$ with respect to $I \subseteq S^+$}
 \label{alg:mixedmulti}
 \end{small}
\end{algorithm}

Here, we show how to compute the mixed distribution $r$ from $p$ and $q$ with a subset $I \subseteq S^+$\,\footnote{An implementation is available at: \url{https://github.com/mahito-sugiyama/information-decomposition}}.
First we present an algorithm to compute $r$ in a simple case, where $I$ is a singleton and we let $I = \{x^*\}$.
Since $\eta_r(x) = \eta_p(x)$ for all $x \not= x^*$, it is clear that $r(x) = p(x)$ for any $x \not\le x^*$.
Therefore, we have focused on computing only $r(x)$ with $x \le x^*$.

Assume $r(x^*)$ is fixed and let $I_{\ge x} = \{\,s \in {\downarrow}x^* \mid s \ge x\,\}$ and $I_{> x} = I_{\ge x} \setminus \{x\}$.
For
each $x \in {\downarrow}x^*$ with $x \not= x^*$, we have $\sum_{s \in I_{\ge x}} p(s) = \sum_{s \in I_{\ge x}} r(s)$ from $\eta_p(x) = \eta_r(x)$.
Hence $r(x)$ is obtained as $r(x) = p(x) + \sum_{s \in I_{> x}} \left(\,p(s) - r(s)\,\right)$.
Thus if ${\downarrow}x^*$ is topologically sorted as $x_0, x_1, \dots, x_{m}$ with $x_0 = \bot$ and $x_m = x^*$,
we can compute $r(x_{m})$, $r(x_{m - 1})$, $\dots$, $r(x_0)$ one after another.
The function {\scshape ComputePSingle}$(x)$ in Algorithm~\ref{alg:mixedsingle} performs for this computation.
Since $\theta_{r}(x_0)$, $\theta_{r}(x_1)$, $\dots$, $\theta_{r}(x_m) = \theta_r(x^*)$ can be computed after computing all $r(x)$ under fixed $r(x^*)$,
$\theta_{r}(x^*)$ is numerically computed as a function of $r(x^*)$.
This process is summarized in the function {\scshape ComputeThetaSingle}$(x^*, r(x^*))$ in Algorithm~\ref{alg:mixedsingle}.
As the function is continuous, we can use a numerical optimization method, such as the bisection method, to efficiently search $r(x^*)$ giving the solution $\theta_{r}(x^*) = \theta_q(x^*)$.
The time complexity of computing $r$ is $O(h(x^*)|{\downarrow}x^*|^2) \le O(h(x^*)|S|^2)$, where $h(x^*)$ is the number of iterations for solving $\theta_{r}(x^*) = \theta_q(x^*)$.

We next consider the general case.
Let $I = \{x_{1}^*, \dots, x_l^*\}$.
Although it is again difficult to analytically compute the mixed distribution $r$, we can numerically compute the distribution $r$ by iterating computation of $\theta_{r}(x_i^*)$ for each $x_i^*$ while fixing $\theta(x_j^*)$ with $j \not= i$,
which is inspired by alternating optimization over $I$ mainly used in the field of convex optimization.
The overall process is shown in Algorithm~\ref{alg:mixedmulti}.
\begin{lemma}
 Algorithm~\ref{alg:mixedmulti} always converges to the mixed distribution $r \in \vec{\Sc}$ of $(p, q)$ with respect to $I \subseteq S^+$.
 \end{lemma}
\begin{proof}
 Let $r_1, r_2, \dots$ be a sequence of distributions in which each $r_i$ is obtained by the $i$th run of the function {\scshape ComputeMixedSingle} in Algorithm~\ref{alg:mixedmulti}.
 From Theorem~\ref{theorem:pythagoras}, we have $\KL(r_i, r) = \KL(r_i, r_{i + 1}) + \KL(r_{i + 1}, r)$ for all $i$, hence $\KL(r_i, r) \ge \KL(r_{i + 1}, r)$ with the equality holding only if $\KL(r_i, r_{i + 1}) = 0$.
 Since there always exists $r_j$ with $j > i$ such that $\KL(r_i, r_j) > 0$ if $\KL(r_i, r) > \epsilon$ for any $\epsilon > 0$, Algorithm~\ref{alg:mixedmulti} converges to the mixed distribution $r$.
 \end{proof}
 Since the time complexity of computing $r(x^*)$ for each $x^* \in I$ is $O(h(x^*)|{\downarrow}x^*|^2)$, the overall time complexity of computing $r$ is $O(h\sum_{x^* \in I}h(x^*)|{\downarrow}x^*|^2) \le O(h|S|^3\sum_{x^* \in I}h(x^*))$, where $h$ is the number of iterations until convergence of $r$.

\begin{full}
\subsection{Metric Space}\label{subsec:distance}
The KL divergence decomposition induces a metric on the poset $S$.
First we assign a real-valued value for each element.
Define a mapping $v : S \to \R^+$ for $x \in S$ as
\begin{align}
 \label{eq:valdef}
 v(x) \defeq \log|S| - \KL(p, p_{{\uparrow}x}),
\end{align}
where $p_{{\uparrow}x}$ is the mixed distribution of $(p, p_0)$ with the uniform distribution $p_0$ with respect to ${\uparrow}x$.
Note that $v(\bot) = H(X)$.

The function $v$ is called \emph{subvaluation}\footnote{The definition in~\cite[Chapter~10.3]{Deza09} applies to lattices, but in our case it can be generalized to posets by naturally setting $v(x \vee y) = v(\emptyset)$ if $x \vee y$ does not exist. See also~\cite{Monjardet81,Orum09} for general cases.}~\cite[Chapter~10.3]{Deza09} if it is strictly isotone, i.e., $x < y$ implies $v(x) < v(y)$, and satisfies the condition:
\begin{align}
 \label{eq:subval}
 v(x \vee y) + v(x \land y) \le v(x) + v(y).
\end{align}
If $v$ is subvaluation, it induces a metric $d_v$ on $S$ defined as
\begin{align*}
 d_v(x, y) \defeq 2v(x \vee y) - v(x) - v(y).
\end{align*}

To show that the mapping $v$ satisfies the above condition~\eqref{eq:subval}, we prepare the following lemma.
\begin{lemma}
 \label{lemma:equality}
 Given two distributions $p, q$, let $p_I$ be the mixed distribution of $(p, q)$ with respect to $I \subseteq S^+$.
 We have
 \begin{align*}
  \KL(p_I, p_{I \cup J}) = \KL(p_{I \cup K}, p_{I \cup J \cup K}),
 \end{align*}
 if $K \cap (J \setminus I) = \emptyset$.
 \end{lemma}
 \begin{proof}
  The difference of $\theta_{p_I}(x)$ and $\theta_{p_{I \cup J}}(x)$ depends on only elements in $J \setminus I$.
  Thus if $K \cap (J \setminus I) = \emptyset$ holds, the difference of $\theta_{p_{I \cup K}}$ and $\theta_{p_{I \cup J \cup K}}$ must be the same as that of $\theta_{p_I}$ and $\theta_{p_{I \cup J}}$.
  Thereby
 \begin{align*}
   \theta_{p_I}(x) - \theta_{p_{I \cup J}}(x) = \theta_{p_{I \cup K}}(x) - \theta_{p_{I \cup J \cup K}}(x).
 \end{align*}
  From Equation~\eqref{eq:loglinear},
  \begin{align*}
   \frac{p_I(x)}{p_{I \cup J}(x)} = \frac{p_{I \cup K}(x)}{p_{I \cup J \cup K}(x)},\ 
   \log \frac{p_I(x)}{p_{I \cup J}(x)} = \log \frac{p_{I \cup K}(x)}{p_{I \cup J \cup K}(x)}.
  \end{align*}
  Therefore $\KL(p_I, p_{I \cup J}) = \KL(p_{I \cup K}, p_{I \cup J \cup K})$ follows.
 \end{proof}

\begin{theorem}[subvaluation]
 The function $v$ in Equation~\eqref{eq:valdef} is subvaluation, that is, it is strictly isotone and it satisfies
 \begin{align*}
  v(x \vee y) + v(x \land y) \le v(x) + v(y).
 \end{align*}
 for all $x, y \in S$.
\end{theorem}
\begin{proof}
 It is clear that $v$ is strictly isotone from Equation~\eqref{eq:valdef}.
 Since $v(x \vee y) + v(x \land y) = v(x) + v(y)$ if $x \le y$ or $y \le x$,
 in the following we assume that $x \not\le y$ and $y \not\le x$.
 From KL divergence decomposition in Theorem~\ref{theorem:pythagoras},
 \begin{align*}
     &  v(x) + v(y) - \left(\,v(x \vee y) + v(x \land y)\,\right)\\
  =\,&  \KL(p, p_{{\uparrow}(x \vee y)}) + \KL(p, p_{{\uparrow}(x \land y)})\\
     &- \KL(p, p_{{\uparrow}x}) - \KL(p, p_{{\uparrow}y})\\
  =\,&  \KL(p, p_{{\uparrow}(x \vee y)}) + \KL(p, p_{{\uparrow}x}) + \KL(p_{{\uparrow}x}, p_{{\uparrow}(x \land y)})\\
     &- \KL(p, p_{{\uparrow}x}) - \KL(p, p_{{\uparrow}(x \vee y)}) - \KL(p_{{\uparrow}(x \vee y)}, p_{{\uparrow}y})\\
  =\,&  \KL(p_{{\uparrow}x}, p_{{\uparrow}(x \land y)}) - \KL(p_{{\uparrow}(x \vee y)}, p_{{\uparrow}y})\\
  =\,&  \KL(p_{{\uparrow}x}, p_{{\uparrow}x \cup {\uparrow}y}) + \KL(p_{{\uparrow}x \cup {\uparrow}y}, p_{{\uparrow}(x \land y)}) - \KL(p_{{\uparrow}(x \vee y)}, p_{{\uparrow}y}).
 \end{align*}
 From Lemma~\ref{lemma:equality},
 \begin{align*}
  \KL(p_{{\uparrow}(x \vee y)}, p_{{\uparrow}y}) = \KL(p_{{\uparrow}x}, p_{{\uparrow}x \cup {\uparrow}y})
 \end{align*}
 since ${\uparrow}x \cap ({\uparrow}y \setminus {\uparrow}(x \vee y)) = \emptyset$.
 Thus we have
 \begin{align*}
  v(x) + v(y) - \left(\,v(x \vee y) + v(x \land y)\,\right) &= \KL(p_{{\uparrow}x \cup {\uparrow}y}, p_{{\uparrow}(x \land y)})\\
  &\ge 0.
 \end{align*}
 It follows that $v$ is subvaluation.
 \end{proof}
From our definition of $v$ and KL divergence decomposition, the distance is given as:
\begin{align*}
 d_v(x, y) &= 2v(x \vee y) - v(x) - v(y)\\
 &= \KL(p, p_{{\uparrow}x}) + \KL(p, p_{{\uparrow}y}) - 2\KL(p, p_{{\uparrow}(x \vee y)})\\
 &= \KL(p_{{\uparrow}(x \vee y)}, p_{{\uparrow}x}) + \KL(p_{{\uparrow}(x \vee y)}, p_{{\uparrow}y}).
\end{align*}
In particular, $d_v(x, y) = \KL(p_{{\uparrow}x}, p_{{\uparrow}y})$ if $x \ge y$.

If we consider the covering graph $G = (V, E)$ of $S$, which is the undirected graph such that $V = S$ and $(x, y) \in E$ if $x$ covers $y$ or $y$ covers $x$,
the KL divergence $\KL(p_{{\uparrow}x}, p_{{\uparrow}y})$ gives the weight for each edge $(x, y)$ and the distance $d(x, y)$ corresponds to the path length from $x$ to $y$ on $G$.
\end{full}

\subsection{Measuring Statistical Significance of $\theta$}\label{subsec:significance}
Given a distribution $p$ on $S$,
we can assess the statistical significance of $\theta_p$ through a likelihood-ratio test, in particular a $G$-test, using decomposition of the KL divergence.
Each $\theta_p(x)$ shows a contribution of $x$ on $p$ as it is the coefficient of the log expansion of $p$ and is orthogonal to the marginals $\eta_p$.

The null and the alternative hypotheses are~\cite{Nakahara02,Nakahara03}:
\begin{align*}
 H_0\!\colon \theta_p(x) = 0, \forall x \in I, \q H_1\!\colon \theta_p(x) \not= 0, \forall x \in I,
\end{align*}
which means that we \emph{knock down} all elements $x \in I$ by letting $\theta_p(x) = 0$ in the generalized log-linear model $\log p(x) = \sum_{s \le x} \theta_p(s)$ in Equation~\eqref{eq:loglinear}.
The statistics $\lambda$ is then given as
\begin{align*}
 \lambda = 2N \sum_{s \in S} \left(\,p(s) \log \left(\,\frac{p(s)}{r(s)}\,\right)\,\right)
 = 2N\KL(p, r),
\end{align*}
where $N$ is the sample size and $r$ is the null distribution, the mixed distribution of $(p, p_0)$ with respect to $I$, and hence $\lambda$ can be computed by Algorithms~\ref{alg:mixedsingle} and~\ref{alg:mixedmulti}.
Therefore, the \emph{$p$-value} can be obtained from data samples since $\lambda$ is known to follow the $\chi^2$-distribution with the degrees of freedom $|S| - 1$.

\begin{full}
\section{Mutual Information Decomposition}\label{sec:MIdecompose}
Let us introduce an additional discrete random variable $Y$ with an arbitrary alphabet $S'$ and consider the mutual information $\MI(X, Y)$ between $X$ and $Y$,
where we do not assume any structure for $S'$.
We show that the mutual information $\MI(X, Y)$ can be also decomposed using the structure of the poset $(S, \le)$ by following the approach proposed in~\cite{Nakahara02}.

Suppose that $p|_y$ is the conditional probability distribution of $p$ given the occurrence of the value $y \in S'$ of $Y$.
It is known that the mutual information $\MI(X, Y)$ can be described using the KL divergence:
\begin{align}
 \label{eq:MIKL}
 \MI(X, Y) &= \sum_{y \in S'} p_Y(y) \KL(p|_y, p) \nonumber\\
 &= \mathbb{E}_{Y}\left[\KL(p|_y, p)\right],
\end{align}
where $p_Y(y) = \Pr(Y = y)$ and $\mathbb{E}_Y$ denotes the expectation of $Y$.
Using KL divergence decomposition (Theorem~\ref{theorem:pythagoras}), for any subset $I \subseteq S$ we can decompose the KL divergence
\begin{align}
 \label{eq:MIKLdecompose}
 \KL(p|_y, p) = \KL(p|_y, p|_{yI}) + \KL(p|_{yI}, p),
 \end{align}
where $p|_{yI}$ is the mixed distribution of $(p|_y, p)$ with respect to $I$, that is,
\begin{align*}
 \left\{
 \begin{array}{ll}
  \eta_{p|_{yI}}(x) \defeq \eta_{p|_y}(x) & \text{if } x \in S^+ \setminus I, \\
  \theta_{p|_{yI}}(x) \defeq \theta_{p}(x) & \text{if } x \in I.
 \end{array}
 \right.
\end{align*}
Thus it is clear from Equations~\eqref{eq:MIKL} and~\eqref{eq:MIKLdecompose} that, for the mutual information $\MI(X, Y)$ and a subset $I \subseteq S$, we have
\begin{align}
 \MI(X, Y) = \mathbb{E}_{Y}\left[\KL(p|_y, p|_{yI})\right] + \mathbb{E}_{Y}\left[\KL(p|_{yI}, p)\right].
\end{align}

We formalize this observation that the mutual information can be \emph{refined} using the information of the structure of $S$ with introducing a new notion of the mutual information on the poset $S$.
\begin{definition}[Refined mutual information]
 \label{def:RMI}
 Given a discrete random variable $X$ with the alphabet $(S, \le)$ and a discrete random variable $Y$ with $S'$.
 For $I, J \subseteq S^+$, define the \emph{refined mutual information} $\RI{X}{Y}{I}{J}$ as
 \begin{align}
  \label{eq:RMI}
  \RI{X}{Y}{I}{J} &\defeq \sum_{y \in S'} p_Y(y) \KL(p|_{yI}, p|_{yJ})\\
  &= \mathbb{E}_{Y}\left[\KL(p|_{yI}, p|_{yJ})\right],
 \end{align}
 where $p|_y$ is the conditional probability of $p$ given $y \in S'$ and $p|_{yI}$ is the mixed distribution of $(p|_y, p)$ with respect to $I$.
\end{definition}
Because the refined mutual information $\RI{X}{Y}{I}{J}$ is the expectation of the KL divergence $\KL(p|_{yI}, p|_{yJ})$, it has the following useful properties:
\begin{enumerate}
 \item $\RI{X}{Y}{I}{J} \ge 0$.
 \item $\RI{X}{Y}{I}{J} = 0$ if $I = J$.
 \item $\RI{X}{Y}{\emptyset}{S^+} = \MI(X, Y)$.
\end{enumerate}
Note that $\RI{X}{Y}{I}{J} \not= \RI{X}{Y}{J}{I}$ in general as the KL divergence $\KL(p|_{yI}, p|_{yJ})$ is not symmetry.

Our main result is nonnegative decomposition of the refined mutual information.
\begin{theorem}[Refined mutual information decomposition]
 \label{theorem:MIKLdecompose}
 For any subsets $I, J, K \subseteq S^+$ satisfying $I \subseteq J \subseteq K$, we have
 \begin{align}
  \RI{X}{Y}{I}{K} = \RI{X}{Y}{I}{J} + \RI{X}{Y}{J}{K}.
 \end{align}
\end{theorem}
\begin{proof}
 It directly follows from Equations~\eqref{eq:MIKL} and~\eqref{eq:MIKLdecompose}.
\end{proof}
\begin{corollary}[Mutual information decomposition]
 For any subset $I \subseteq S^+$, we have
 \begin{align}
  \MI(X, Y) = \RI{X}{Y}{\emptyset}{I} + \RI{X}{Y}{I}{S^+}.
 \end{align}
 Moreover, for a hierarchical collection $\{I_0, I_1, \dots, I_k\}$ such that $\emptyset = I_0 \subseteq I_1 \subseteq \dots \subseteq I_{k} = S^+$, we have
 \begin{align}
  \MI(X, Y) = \sum_{i = 1}^k \RI{X}{Y}{I_{i - 1}}{I_i}.
 \end{align}
 \end{corollary}
These results show that the structure of $S$ makes it possible to refine and decompose the mutual information $\MI(X, Y)$.
We can use Algorithms~\ref{alg:mixedsingle} and~\ref{alg:mixedmulti} to compute $p|_{yI}$ by letting $p = p|_y$, $q = p$, and $r = p|_{yI}$.

Since $\MI(X, Y)$ intuitively means the strength of the association between $X$ and $Y$, the refined value $\RI{X}{Y}{I}{J}$ shows the contribution of $J \setminus I$. 
In particular, we can isolate any subset $I \subseteq S$ from the total mutual information $\MI(X, Y)$ and measure the degree of association of $I$ by the refined mutual information $\RI{X}{Y}{\emptyset}{I}$, in which $\RI{X}{Y}{\emptyset}{I} = 0$ holds if $\theta_{p|_{yI}}(x) = \theta_{p}(x)$ for all $x \in I$.

\begin{figure}[t]
 \centering
 \includegraphics{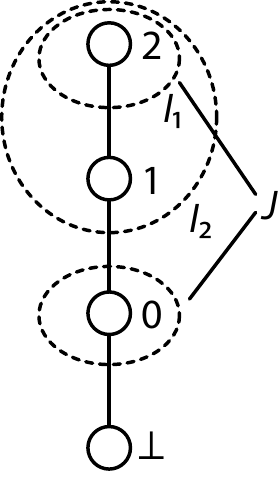}
 \caption{Poset $S$ in Example~\ref{example:RI}.}
 \label{figure:RI_eg}
\end{figure}

\begin{example}
 \label{example:RI}
 Let $S = \{0, 1, 2, 3\}$ and $S' = \{0, 1\}$ for random variables $X$ and $Y$, respectively, and assume a simple structure such that the order is given as $0 \le 1 \le 2 \le 3$ for $S$ as shown in Figure~\ref{figure:RI_eg}.
 Suppose that we have the following probabilities:
 \begin{center}
 \begin{tabular}[t]{cc|cccc|c}
  && \multicolumn{4}{c|}{$X$} & \multirow{2}{20pt}{Total}\\
  && $0$ & $1$ & $2$ & $3$  & \\
  \midrule
  \multirow{2}{0pt}{$Y$} & $0$ & $0.01$ & $0.30$ & $0.10$ & $0.02$ & $0.43$\\
  & $1$ & $0.10$ & $0.13$ & $0.14$ & $0.20$ & $0.57$\\
  \midrule
  \multicolumn{2}{c|}{Total} & $0.11$ & $0.43$ & $0.24$ & $0.22$ & $1.00$ \\
 \end{tabular}
 \end{center}
 The mutual information is as follows.
 \begin{align*}
  \MI(X, Y) &= H(X) + H(Y) - H(X, Y)\\
  &= 1.281 + 0.6833 - 1.808 = 0.1562.
 \end{align*}
 For simplicity, we denote $f = (f(0), f(1), f(2), f(3))$ to show each value $f(x)$ of a function $f$.
 Let $I = \{2, 3\}$.
 From $\theta_{p|_0}$, $\theta_{p|_1}$, $\theta_p$, $\eta_{p|_0}$, $\eta_{p|_1}$, and $\eta_p$,
 The probability distribution $p|_{yI}$ is computed as $p|_{0I} = (0.0233, 0.472, 0.263, 0.241)$ and $p|_{1I} = (0.175, 0.398, 0.222, 0.204)$.
 Thus we have
 \begin{align*}
  \MI(X, Y) &= \RI{X}{Y}{\emptyset}{I} + \RI{X}{Y}{I}{S^+}\\
  &= 0.1219 + 0.0343.
 \end{align*}
 Moreover, $\RI{X}{Y}{\emptyset}{\{1\}} = 0.0713$, $\RI{X}{Y}{\emptyset}{\{2\}} = 0.0252$, $\RI{X}{Y}{\emptyset}{\{3\}} = 0.0340$.
 Thus the element $1 \in S$ has the most strong association with $Y$, $3 \in S$ is the second, and $2 \in S$ is the third.
\end{example}
\end{full}

\begin{full}
\section{Structure Learning from Data}\label{sec:learning}
 We illustrate how to obtain the partial order structure of $S$ from datasets in three case studies, where a dataset is a set of binary vectors, natural vectors, or real vectors.
 In particular, the first case (Section~\ref{subsec:binary}) corresponds to analyzing interactions of $n$ events $e_1, e_2, \dots, e_n$, where each event is treated as a binary random variable,
 It also applies whenever we use dummy coding of categorical variables.

\subsection{Orthogonal Decomposition of Interactions}\label{subsec:binary}
\end{full}
\begin{ISIT}
\section{Orthogonal Decomposition of Interactions}\label{sec:learning}
\end{ISIT}
As a representative application, let us consider the problem of orthogonal decomposition of event combinations.
Suppose there are $n$ events $e_1, \dots, e_n$ as discussed in the Introduction.
For each subset $x \subseteq \sen = \{1, \dots, n\}$, let $p(x)$ be the probability of the combination $\bigcap_{i \in x} e_i$.
The objective is to decompose $\log p(x)$ to the sum of coefficients of its subsets $s \subseteq x$, which correspond to the $\theta$-coordinates:
$\log p(x) = \sum_{s \le x} \theta(s)$.
The order $\le$ is given according to the inclusion relationship: $x \le s$ if $x \subseteq s$.
The coefficients $\theta(s)$ show the ``pure'' contributions of respective interactions $\bigcap_{j \in s} e_j$ as they are independent of their frequencies; that is, the $\eta$-coordinates: $\eta(s) = \sum_{x \ge s} p(x)$


Assume that $N$ samples $t_1, t_2, \dots, t_N$ are given, where each sample $t_i$ is a set of events, which means that the events occur simultaneously.
We estimate each probability $p(x)$ through its natural estimator $\hat{p}(x) = |\{i \in [N] \mid t_i = x \}| \,/\, N$.
To effectively estimate $\hat{p}$ and efficiently compute $\theta_{\hat{p}}$ and $\eta_{\hat{p}}$ from samples,
we prune the whole event combinations $\Pc(\sen)$ by excluding combinations that do not frequently appear in the dataset.
Given a threshold $\sigma \in \R$ such that $0 \le \sigma \le 1$, we set
$S^+ = \{\,x \subseteq \sen \mid \hat{p}(x) \ge \sigma\,\}$
and
$\hat{p}(\bot) = 1 - \sum_{x \in S^+} \hat{p}(x)$.
Thus the dimensionality of the manifold $\vec{\Sc}$ reduces from $2^n$ to at most $N$.
Since any subset of $\Pc(\sen)$ is a poset, we can apply our decomposition technique presented in Section~\ref{sec:probPosets} via computation of $\theta_{\hat{p}}$, $\eta_{\hat{p}}$, and mixed distributions.
Interestingly, a sample $t_i$ can be viewed as a transaction of a database and $\eta_{\hat{p}}(I)$ corresponds to the \emph{support} of $I$ used in the context of frequent pattern (itemset) mining~\cite{Aggarwal14FPM}.


\begin{figure}
 \begin{minipage}[c]{.4\linewidth}
  \centering
  \captionof{table}{Samples in Example~\ref{example:eg1}.}
  \label{table:eg1}
  \begin{tabularx}{\linewidth}{l|l}
   \toprule
   & Events\\
   \midrule
   $t_1$ & $e_2$\\
   $t_2$ & $e_2$\\
   $t_3$ & $e_4$, $e_5$\\
   $t_4$ & $e_1$, $e_2$, $e_4$, $e_5$\\
   $t_5$ & $e_1$, $e_2$, $e_4$, $e_5$\\
   $t_6$ & $e_3$\\
   $t_7$ & $e_1$, $e_2$, $e_4$, $e_5$\\
   $t_8$ & $e_4$, $e_5$\\
   $t_9$ & $e_1$, $e_2$, $e_4$, $e_5$\\
   $t_{10}$ & $e_2$\\
   \bottomrule
  \end{tabularx}
 \end{minipage}\hfill
 \begin{minipage}[c]{.58\linewidth}
  \centering
  \includegraphics[width=.8\linewidth]{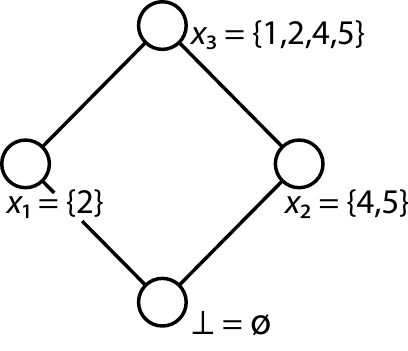}
  \caption{Poset generated from samples in Table~\ref{table:eg1}.}
  \label{figure:eg1}
 \end{minipage}
\end{figure}

\begin{example}
 \label{example:eg1}
 Given samples in Table~\ref{table:eg1}, assume that our threshold $\sigma = 0.2$.
 We then obtain a poset $S = \{\bot, x_1, x_2, x_3\}$ with $\bot = \emptyset$, $x_1 = \{2\}$, $x_2 = \{4, 5\}$, and $x_3 = \{1, 2, 4, 5\}$,
 as shown in Figure~\ref{figure:eg1}, where
 $\hat{p}(\bot) = 0.1$, $\hat{p}(x_1) = 0.3$, $\hat{p}(x_2) = 0.2$, and $\hat{p}(x_3) = 0.4$.
 Thus, $\theta_{\hat{p}}$ are obtained as follows:
 $\theta_{\hat{p}}(\bot) = -2.303$, $\theta_{\hat{p}}(x_1) = 1.099$, $\theta_{\hat{p}}(x_2) = 0.693$, and $\theta_{\hat{p}}(x_3) = -0.405$.
 Let $\hat{r}_x$ be the mixed distribution of $(\hat{p}, p_0)$ with $\{x\} \in S$.
 From these parameters, the KL divergence for each interaction is obtained as $\KL(\hat{p}, \hat{r}_{x_1}) = 0.0523$, $\KL(\hat{p}, \hat{r}_{x_2}) = 0.0170$, and $\KL(\hat{p}, \hat{r}_{x_3}) = 0.0040$.
 Although $p$-values of those interactions are larger than $0.99$, they are due to small sample size $N = 10$.
 If $N = 300$, for example, the $p$-value of $x_1$ becomes $0.001$ and it is significant under the significance level $\alpha = 0.05$.
 \end{example}

\begin{full}
\subsection{Nonnegative Integer Vectors}\label{subsec:natural}
\end{full}
The same strategy can be applied to a poset $S$ composed of vectors of $n$-dimensional nonnegative integers $\Z^n_{\ge 0}$.
We assume $S$ to be a subset of $\Z^n_{\ge 0}$, where for each pair of vectors $x, y \in \Z^n_{\ge 0}$ with $\vec{x} = (x^1, x^2, \dots, x^n)$ and $\vec{y} = (y^1, y^2, \dots, y^n)$,
we define the partial order $\le$ as $\vec{x} \le \vec{y}$ if and only if $x^i \le y^i$ for all $i \in \sen$, and $0 = (0, 0, \dots, 0) \in \Z^n_{\ge 0}$ corresponds to $\bot$.
Any subset $S \subset \Z^n_{\ge 0}$ becomes a poset.

Given $N$ data points $\vec{x}_1, \vec{x}_2, \dots, \vec{x}_N$ of $n$-dimensional nonnegative integers.
Similar to the previous case, a poset $S^+$ is obtained from data as
$S^+ = \{\,\vec{x} \in \Z^n_{\ge 0} \mid \hat{p}(\vec{x}) \ge \sigma\,\}$ using a threshold $\sigma \in \R$.
We can apply our information decomposition to $S$ with an empirical probability distribution $\hat{p}$.

\begin{example}
 \label{example:nat}
 Given data points $\vec{x}_1, \vec{x}_2, \dots, \vec{x}_{25} \in \Z^2_{\ge 0}$ as
 $\vec{x}_1, \dots, \vec{x}_3 = (0, 1)$, $\vec{x}_4 = (1, 0)$,
 $\vec{x}_5, \dots, \vec{x}_8 = (1, 1)$, $\vec{x}_9, \dots, \vec{x}_{11} = (1, 2)$,
 $\vec{x}_{12}, \dots, \vec{x}_{21} = (2, 1)$, $\vec{x}_{22}, \dots, \vec{x}_{25} = (3, 3)$.
 We have $S = \{(0, 0), (0, 1), (1, 1), (1, 2), (2, 1), (3, 3)\}$ if $\sigma = 2/25$, which is shown in Figure~\ref{figure:nat_eg}.
\end{example}

\begin{figure}[t]
 \centering
 \includegraphics[width=.47\linewidth]{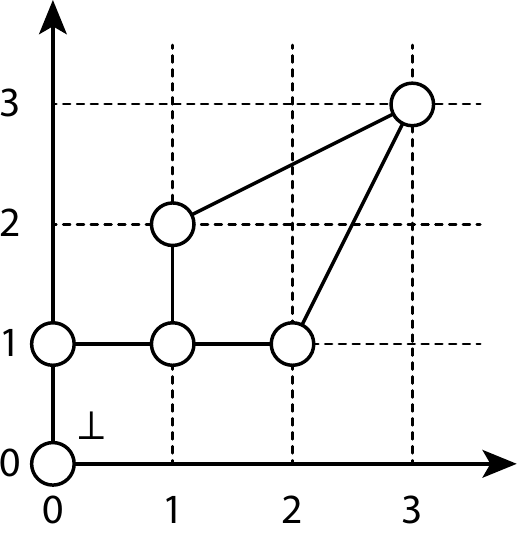}
 \caption{Poset of nonnegative integer vectors in Example~\ref{example:nat}.}
 \label{figure:nat_eg}
\end{figure}

\begin{full}
\subsection{Real Vectors}\label{subsec:real}
Since any subset $S$ of the set $\R^n$ of $n$-dimensional real vectors is a poset with the same partial order as that of $\Z^n_{\ge 0}$,
we can again apply our information decomposition to $S \subset \R^n$ as long as $S$ is finite and the least element $\bot \in S$.
However, given $N$ data points $\vec{x}_1, \vec{x}_2, \dots, \vec{x}_N$ of $n$-dimensional real vectors in $\R^n$,
the problem is how to estimate probabilities $\hat{p}(\vec{x}_i)$ as $\vec{x}_i \not= \vec{x}_j$ in general and the natural estimator is no longer effective for estimation.

A representative approach to solve the problem is partitioning/discretizing a given dataset beforehand by clustering.
Let a dataset $V = \{\vec{x}_1, \vec{x}_2, \dots, \vec{x}_N\}$ and $V_1, V_2, \dots, V_k$ be $k$ clusters of $V$ such that
$V_i \cap V_j = \emptyset$ for all $i, j \in [k]$ and $\bigcup_{i \in [k]} V_i  = V$.
We assume that each cluster $V_i$ has its representative point $\vec{c}_i \in \R^n$.
Then probability is estimated as $\hat{p}(\vec{c}_i) = |V_i| / N$ for each $i \in [k]$ and
\begin{align*}
 S^+ = \{\,\vec{c}_i \in \R^n \mid \hat{p}(\vec{c}_i) \ge \sigma\,\}
\end{align*}
with $\hat{p}(\bot) = 1 - \sum_{c \in S^+} \hat{p}(\vec{c})$.

We can use any clustering methods as long as they give a representative point for each cluster.
Examples include the $K$-means algorithm~\cite{Macqueen67}, where $\vec{c}_i$ is the center of the cluster $V_i$, and
the $K$-medoids algorithm~\cite{Kaufman90}, where $\vec{c}_i$ is a data point $\vec{x} \in V_i$ that minimizes the sum of distances $\sum_{\vec{y} \in V_i} D(\vec{x}, \vec{y})$.
See~\cite{Aggarwal13Clustering} for reviews of other clustering methods.
\end{full}

\section{Conclusion}\label{sec:conclusion}
In this paper, we have theoretically shown the intriguing relationship between two key structures in information geometry and order theory:
the dually flat structure of a manifold of the exponential family and the partial order structure of events.
We have proposed an efficient algorithm of information decomposition that is applicable to any kind of posets; this is in contrast to a number of other studies~\cite{Bertschinger13,Bertschinger14,Olbrich15,Williams10}.

As a representative application, we have demonstrated orthogonal decomposition of interactions of events.
We have shown that the partial order structure can be directly obtained from data in an efficient manner, and
the dimensionality of the manifold is reduced from $2^n$ for $n$ variables in previous approaches to, at most, the sample size $N$.
Thus, we can perform orthogonal decomposition for recently emerging high-dimensional data with thousands or even millions of variables, such as single nucleotide polymorphisms (SNPs) in genome-wide association studies (GWAS)~\cite{Wellcome07} and neural data in neuroscience~\cite{Alivisatos12}.
To our knowledge, this is the first method that avoids the curse of dimensionality in orthogonal decomposition of interactions and achieves efficient computation and effective probability estimation from data.

Our work promises many interesting future studies, both in theoretical and practical directions.
There will be a more interesting theoretical connection between information geometry and order theory.
Furthermore, it is exciting to apply our decomposition method to real-world scientific datasets such as firing patterns of neurons and SNPs in GWAS to reveal unknown associations.


%



\section*{Acknowledgment}
This work was supported by JSPS KAKENHI Grant Number 26880013 (MS) and 26120732 (HN).
The research of K.T.~was supported by JST CREST, JST ERATO, RIKEN PostK, NIMS MI2I, KAKENHI Nanostructure and KAKENHI 15H05711.



%

\bibliographystyle{IEEEtran}
\bibliography{main}

\begin{thebibliography}{10}
\providecommand{\url}[1]{#1}
\csname url@samestyle\endcsname
\providecommand{\newblock}{\relax}
\providecommand{\bibinfo}[2]{#2}
\providecommand{\BIBentrySTDinterwordspacing}{\spaceskip=0pt\relax}
\providecommand{\BIBentryALTinterwordstretchfactor}{4}
\providecommand{\BIBentryALTinterwordspacing}{\spaceskip=\fontdimen2\font plus
\BIBentryALTinterwordstretchfactor\fontdimen3\font minus
  \fontdimen4\font\relax}
\providecommand{\BIBforeignlanguage}[2]{{%
\expandafter\ifx\csname l@#1\endcsname\relax
\typeout{** WARNING: IEEEtran.bst: No hyphenation pattern has been}%
\typeout{** loaded for the language `#1'. Using the pattern for}%
\typeout{** the default language instead.}%
\else
\language=\csname l@#1\endcsname
\fi
#2}}
\providecommand{\BIBdecl}{\relax}
\BIBdecl

\bibitem{Amari01}
S.~Amari, ``Information geometry on hierarchy of probability distributions,''
  \emph{IEEE Transactions on Information Theory}, vol.~47, no.~5, pp.
  1701--1711, 2001.

\bibitem{Nakahara02}
H.~Nakahara and S.~Amari, ``Information-geometric measure for neural spikes,''
  \emph{Neural Computation}, vol.~14, no.~10, pp. 2269--2316, 2002.

\bibitem{Nakahara06}
H.~Nakahara, S.~Amari, and B.~J. Richmond, ``A comparison of descriptive models
  of a single spike train by information-geometric measure,'' \emph{Neural
  computation}, vol.~18, no.~3, pp. 545--568, 2006.

\bibitem{Nakahara03}
H.~Nakahara, S.~Nishimura, M.~Inoue, G.~Hori, and S.~Amari, ``Gene interaction
  in {DNA} microarray data is decomposed by information geometric measure,''
  \emph{Bioinformatics}, vol.~19, no.~9, pp. 1124--1131, 2003.

\bibitem{Hou13}
Y.~Hou, X.~Zhao, D.~Song, and W.~Li, ``Mining pure high-order word associations
  via information geometry for information retrieval,'' \emph{ACM Transactions
  on Information Systems}, vol.~31, no.~3, pp. 12:1--12:32, 2013.

\bibitem{Ganmor11}
E.~Ganmor, R.~Segev, and E.~Schneidman, ``Sparse low-order interaction network
  underlies a highly correlated and learnable neural population code,''
  \emph{Proceedings of the National Academy of Sciences}, vol. 108, no.~23, pp.
  9679--9684, 2011.

\bibitem{Amari07}
S.~Amari and H.~Nagaoka, \emph{Methods of information geometry}.\hskip 1em plus
  0.5em minus 0.4em\relax American Mathematical Society, 2007.

\bibitem{Davey02}
B.~A. Davey and H.~A. Priestley, \emph{Introduction to lattices and order},
  2nd~ed.\hskip 1em plus 0.5em minus 0.4em\relax Cambridge University Press,
  2002.

\bibitem{Gierz03}
G.~Gierz, K.~H. Hofmann, K.~Keimel, J.~D. Lawson, M.~Mislove, and D.~S. Scott,
  \emph{Continuous Lattices and Domains}.\hskip 1em plus 0.5em minus
  0.4em\relax Cambridge University Press, 2003.

\bibitem{Aggarwal14FPM}
C.~C. Aggarwal and J.~Han, Eds., \emph{Frequent Pattern Mining}.\hskip 1em plus
  0.5em minus 0.4em\relax Springer, 2014.

\bibitem{Bertschinger13}
N.~Bertschinger, J.~Rauh, E.~Olbrich, and J.~Jost, ``Shared information---new
  insights and problems in decomposing information in complex systems,'' in
  \emph{Proceedings of the European Conference on Complex Systems 2012}.\hskip
  1em plus 0.5em minus 0.4em\relax Springer, 2013, pp. 251--269.

\bibitem{Bertschinger14}
N.~Bertschinger, J.~Rauh, E.~Olbrich, J.~Jost, and N.~Ay, ``Quantifying unique
  information,'' \emph{Entropy}, vol.~16, no.~4, pp. 2161--2183, 2014.

\bibitem{Olbrich15}
E.~Olbrich, N.~Bertschinger, and J.~Rauh, ``Information decomposition and
  synergy,'' \emph{Entropy}, vol.~17, no.~5, pp. 3501--3517, 2015.

\bibitem{Williams10}
P.~L. Williams and R.~D. Beer, ``Nonnegative decomposition of multivariate
  information,'' \emph{arXiv:1004.2515}, 2010.

\bibitem{Wellcome07}
{The Wellcome Trust Case Control Consortium}, ``Genome-wide association study
  of 14,000 cases of seven common diseases and 3,000 shared controls,''
  \emph{Nature}, vol. 447, no. 7145, pp. 661--678, 2007.

\bibitem{Alivisatos12}
A.~Alivisatos, M.~Chun, G.~Church, R.~Greenspan, M.~Roukes, and R.~Yuste, ``The
  brain activity map project and the challenge of functional connectomics,''
  \emph{Neuron}, vol.~74, pp. 970--974, 2012.

\end{thebibliography}

\end{document}